\documentclass[12pt,english]{article}
\usepackage[latin9]{inputenc}
\usepackage{geometry}
\usepackage[active]{srcltx}
\usepackage{color}
\usepackage{babel}
\usepackage{amsthm}
\usepackage{amsmath}
\usepackage[unicode=true,pdfusetitle,
 bookmarks=true,bookmarksnumbered=false,bookmarksopen=false,
 breaklinks=false,pdfborder={0 0 1},backref=false,colorlinks=true]
 {hyperref}

\makeatletter

\providecommand{\tabularnewline}{\\}

\numberwithin{equation}{section}
\theoremstyle{plain}
\newtheorem{thm}{\protect\theoremname}[section]
  \theoremstyle{definition}
  \newtheorem{defn}[thm]{\protect\definitionname}
  \theoremstyle{plain}
  \newtheorem{cor}[thm]{\protect\corollaryname}
  \theoremstyle{plain}
  \newtheorem{algorithm}[thm]{\protect\algorithmname}
  \theoremstyle{definition}
  \newtheorem{example}[thm]{\protect\examplename}

\@ifundefined{date}{}{\date{}}
\makeatother

  \providecommand{\algorithmname}{Algorithm}
  \providecommand{\corollaryname}{Corollary}
  \providecommand{\definitionname}{Definition}
  \providecommand{\examplename}{Example}
\providecommand{\theoremname}{Theorem}

\begin{document}

\title{Fitting log-linear models in sparse contingency tables using the
eMLEloglin R package}

\author{Matthew Friedlander}

\maketitle
\begin{abstract}
Log-linear modeling is a popular method for the analysis of contingency
table data. When the table is sparse, the data can fall on the boundary
of the convex support, and we say that the MLE does not exist in the
sense that some parameters cannot be estimated. However, an extended
MLE always exists, and a subset of the original parameters will be
estimable. The eMLEloglin package determines which sampling zeros
contribute to the non-existence of the MLE. These problematic cells
can be removed from the contingency table and the model can then be
fit (as far as is possible) using the $\mathtt{glm()}$ function. 
\end{abstract}

\section{Introduction}

Data in the form of a contingency table arise when individuals are
cross classified according to a finite number of criteria. Log-linear
modeling (see e.g., \cite{Agresti}, \cite{Christensen}, or \cite{Bishop})
is a popular and effective methodology for analyzing such data enabling
the practitioner to make inferences about dependencies between the
various criteria. For hierarchical log-linear models, the interactions
between the criteria can be represented in the form of a graph; the
vertices represent the criteria and the presence or absence of an
edge between two criteria indicates whether or not the two are conditionally
independent \cite{Lauritzen}. This kind of graphical summary greatly
facilitates the interpretation of a given model. 

Log-linear models are typically fit by maximum likelihood estimation
(i.e. we attempt compute the MLE of the expected cell counts and log-linear
parameters). It has been known for some time that problems arise when
the sufficient statistic falls on the boundary of the convex support,
say $C$, of the model. This generally occurs in sparse contingency
tables with many zeros. In such cases, algorithms for computing the
MLE can fail to converge or become unstable. Moreover, the effective
model dimension will be reduced and the degrees of freedom of the
usual goodness of fit statistics will be incorrect \cite{Feinberg1}.
Only fairly recently, have algorithms been developed to begin to deal
with this situation (\cite{Eriksson}, \cite{Geyer} and \cite{Feinberg2}).
It turns out that identification of the face $F$ of $C$ containing
the data in its relative interior is crucial to efficient and reliable
computation of the MLE and the effective model dimension. If $F=C$,
then the MLE exists and it's calculation is straightforward. If not
(i.e. $F\subset C$), the log-likelihood has its maximum on the boundary
and remedial steps must be taken to find and compute those parameters
that can be estimated.

The outline of this paper is as follows. In section 2, we describe
necessary and sufficient conditions for the existence of the MLE.
In section 3, we place these conditions in the context of convex geometry.
In section 4, we describe a linear programming algorithm to find $F$.
We then discuss how to compute ML estimates and find the effective
model dimension. In section 5, we introduce the eMLEloglin R package
for carrying out the tasks described in section 4.

\section{Conditions for the existence of the MLE}

Let $V$ be a finite set of indices representing $|V|$ criteria.
We assume that the criterion labeled by $v\in V$ can take values
in a finite set $\mathcal{I}_{v}$. The resulting counts are gathered
in a contingency table such that 
\[
\mathit{\mathcal{I}=}\prod_{v\in V}\mathcal{I}_{v}
\]
is the set of cells $i=\left(i_{v},v\in V\right)$. The vector of
cell counts is denoted $n=\left(n(i),i\in\mathcal{I}\right)$ with
corresponding mean $m(i)=E(n)=\left(m(i),i\in\mathcal{I}\right)$.
For $D\subset V,$
\[
\mathcal{I}_{D}=\prod_{v\in D}\mathcal{I}_{v}
\]
is the set of cells $i_{D}=(i_{v},v\in D)$ in the $D$-marginal table.
The marginal counts are $n(i_{D})=\sum_{j:j_{D}=i_{D}}n(j)$ with
$m(i_{D})=E\left(n(i_{D})\right)$. 

We assume that the components of $n$ are independent and follow a
Poisson distribution (i.e. Poisson sampling) and that the cell means
are modeled according to a hierarchical model 
\[
\log\left(m\right)=X\theta
\]
where $X$ is an $|\mathcal{I}|\times d$ design matrix of full column
rank with rows $\left\{ f_{i},i\in\mathcal{I}\right\} $ and $\theta$
is a vector of log-linear parameters with $\theta\in R^{d}$. The
results herein also apply under multinomial or product multinomial
sampling. We assume that the first component of $f_{i}$ is 1 for
all $i\in\mathcal{I}$ and that ``baseline constraints'' are used
making the $f_{i}'s$ binary 0/1 vectors. For each cell, $i\in\mathcal{I}$,
we have $\log m(i)=\left\langle f_{i},\theta\right\rangle $. 

The sufficient statistic $t=X^{T}n$ has a probability distribution
in the natural exponential family
\[
f(t)=\exp\left(\left\langle \theta,t\right\rangle -\sum_{i\in\mathcal{I}}\exp\left(\left\langle f_{i},\theta\right\rangle \right)\right)\nu\left(dt\right)
\]
with respect to a discrete measure $\nu$ that has convex support
\[
C^{p}=\left\{ \sum_{i\in\mathcal{I}}y(i)f_{i},y(i)\ge0,i\in\mathcal{I}\right\} =\mathrm{cone}\left\{ f_{i},i\in\mathcal{I}\right\} 
\]
i.e. the convex cone generated by the rows of the design matrix $X$.
The log-likelihood, as a function of $m,$ is
\[
l(m)=\sum_{i\in\mathcal{I}}n(i)\log m(i)-\sum_{i\in\mathcal{I}}m(i)
\]
Let $\mathcal{M}$ be the column space of $X$. It is well known that
the log-likelihood is strictly concave with a unique maximizer $\hat{m}=\mathrm{argsup}_{\log m\in\mathcal{M}}l(m)$
that satisfies $X^{T}\hat{m}=t$. 
\begin{defn}
\label{def:2.1} If $\hat{m}>0$, we say that $\hat{m}$ is the MLE
of $m$ while if $\hat{m}(i)=0$ for some $i\in\mathcal{I}$ we call
$\hat{m}$ the extended MLE. 
\end{defn}
The following important theorem from \cite{Haberman} gives necessary
and sufficient conditions for $\hat{m}>0$, i.e. the existence of
the MLE.
\begin{thm}
\label{thm:2.2} The MLE exists if and only if there exists a $y$
such that $X^{T}y=0$ and $y+n>0$. \end{thm}
\begin{proof}
Suppose that $\hat{m}$ exists. Then $X^{T}\hat{m}=t$ and $X^{T}(\hat{m}-n)=0.$
Letting $y=\hat{m}-n$ we have $X^{T}y=0$ and $y+n=\hat{m}>0$. 
\end{proof}
Conversely, suppose that there exists a $y$ such that $X^{T}y=0$
and $y+n>0$. Then $\sum_{i\in\mathcal{I}}y(i)\log m(i)=\sum_{i\in\mathcal{I}}y(i)\left\langle f_{i},\theta\right\rangle =\left\langle \sum_{i\in\mathcal{I}}y(i)f_{i},\theta\right\rangle =\left\langle X^{T}y,\theta\right\rangle =0$.
We can then write the log-likelihood as 
\begin{eqnarray*}
l(m) & = & \sum_{i\in\mathcal{I}}n(i)\log m(i)-\sum_{i\in\mathcal{I}}m(i)\\
 & = & \sum_{i\in\mathcal{I}}\left(y(i)+n(i)\right)\log m(i)-\sum_{i\in\mathcal{I}}m(i)
\end{eqnarray*}
Let $\mu=\log m$ and consider the real valued function $f\left(\mu(i)\right)=\left(y(i)+n(i)\right)\mu(i)-\exp\left(\mu(i)\right)$
for some $i\in\mathcal{I}$. Differentiating with respect to $\mu(i)$,
we have $f'\left(\mu(i)\right)=\left(y(i)+n(i)\right)-\exp\left(\mu(i)\right)$
and $f''(\mu(i))=-\exp\left(\mu(i)\right)<0$, and we see that $f$
is strictly concave with a finite maximum $\mu(i)=\log\left(y(i)+n(i)\right)$
and $l$ is bounded above. $l$ is not bounded below, however, since
$\lim_{\mu(i)\rightarrow\pm\infty}$$f\left(\mu(i)\right)=-\infty$.

Let $A$ be the set $\left\{ \mu\in\mathcal{M}:l\left(\mu\right)\ge c\right\} $
where $c\in R$. The number $c$ can be chosen small enough such that
the set $A$ is non-empty. Then $A$ is bounded and, since $l$ is
continuous function $\mu$, it is closed. It follows that $A$ is
compact and $l$ must have a finite maximum, $\hat{\mu}$ for some
$\mu\in A$. We conclude that $\hat{m}>0$.
\begin{cor}
\label{cor:2.3} If $n>0$, the MLE exists.\end{cor}
\begin{proof}
Take $y=0$ in theorem \ref{thm:2.2}.
\end{proof}

\section{Some basics of convex geometry}

In this section we give some basic definitions that we need from convex
geometry. Some supplementary references are \cite{Rockafellar} and
\cite{Ziegler}. In general, a polytope is a closed object with flat
sides. The relative interior of a polytope is its interior with respect
to the affine space of smallest dimension containing it. For a polytope
that is full dimensional, the relative interior corresponds the the
topological interior (int). The cone, generated by the points $a_{1},a_{2},...,a_{n}$
is a polytope given by 
\[
\mathrm{cone}\left\{ a_{1},a_{2},...,a_{n}\right\} =\left\{ \sum_{i=1}^{n}a_{i}x_{i}:x_{i}\ge0,i=1,2,...,n\right\} 
\]
and its relative interior is
\[
\left\{ \sum_{i=1}^{n}\lambda_{i}a_{i}:\lambda_{i}>0,i=1,2,...,n\right\} .
\]
The convex hull of the same set of points, $\mathrm{conv}\left\{ a_{1},a_{2},...,a_{n}\right\} $,
is also a polytope with the added restriction that $\sum_{i=1}^{n}\lambda_{i}=1$.
A convex polytope $P$ can be represented as the convex hull of a
finite number of points (the V-representation) or, equivalently, as
the intersection of a finite number of half spaces (the H-representation).
Cones and convex hulls are examples of convex polytopes. Henceforth,
we assume that $P$ is a convex polytope in $R^{d}$. 

A face of $P$ is a nonempty set of the form $F=P\cap\left\{ x\in R^{d}:c^{T}x=r\right\} $
where $c^{T}x\le r$ for all $x\in P$. The set $\left\{ x\in R^{d}:c^{T}x=r\right\} $
is called a supporting hyperplane to $P$. The faces of dimension
0 are called extreme points and, if $P$ is a cone, the one dimensional
faces of $P$ are called the extreme rays of $P$. Moreover, when
$P$ is a cone all faces include the origin so that $r=0$ and the
origin is the only face of dimension 0. The dimension of a face $F$
is the dimension of its affine hull
\[
\mathrm{aff}\left\{ \sum_{i}\lambda_{i}x_{i}:x_{i}\in F,\sum_{i}\lambda_{i}=1\right\} 
\]
which is the set of all affine combinations of the points in $F$.
Finally, note that by taking $c=0$, $P$ itself is a face. We now
have a sequence of important theorems. .
\begin{thm}
\label{thm:3.1} Any face of $C^{p}$ of dimension at least one is
the cone generated by the $f_{i}'s$ that belong to that face.\end{thm}
\begin{proof}
Suppose that $t$ belongs to a face $F=C^{p}\cap\left\{ x\in R^{J}:c^{T}x=0\right\} $
of $C_{p}$ of dimension at least one. Then $F$ contains at least
one point other than the origin. Let $\mathcal{\mathcal{I}}_{F}=\left\{ i\in\mathcal{I}:f_{i}\in F\right\} $.
Every point in $C^{p}$ can be expressed as a conical combination
of the $f_{i}'s$ and, hence, there exist non-negative real numbers
$\left(\lambda_{i},i\in\mathcal{I}\right)$ such that $t=\sum_{i\in\mathcal{I}}\lambda_{i}f_{i}=\sum_{i\in\mathcal{I}_{F}}\lambda_{i}f_{i}+\sum_{i\in\mathcal{I}\backslash\mathcal{I}_{F}}\lambda_{i}f_{i}$.
If $t=0$, then since the first coordinate of each $f_{i}$ is 1,
we must have $\lambda_{i}=0$ for all $i\in\mathcal{I}$ and we can
certainly write $t=\sum_{i\in\mathcal{I}_{F}}\lambda_{i}f_{i}$. If
$t\ne0$ then there must be an $i\in\mathcal{I}$ such that $\lambda_{i}>0$.
Suppose that $\lambda_{i}>0$ for some $i\in\mathcal{I\backslash\mathcal{I}}_{F}$.
Then 
\begin{eqnarray*}
0 & = & c^{T}t=c^{T}\left(\sum_{i\in\mathcal{I}_{F}}\lambda_{i}f_{i}+\sum_{i\in\mathcal{I}\backslash\mathcal{I}_{F}}\lambda_{i}f_{i}\right)\\
 & = & \sum_{i\in\mathcal{I}_{F}}\lambda_{i}\left(c^{T}f_{i}\right)+\sum_{i\in\mathcal{I}\backslash\mathcal{I}_{F}}\lambda_{i}\left(c^{T}f_{i}\right)\\
 & = & \sum_{i\in\mathcal{I}\backslash\mathcal{I}_{F}}\lambda_{i}\left(c^{T}f_{i}\right)\\
 & < & 0
\end{eqnarray*}
and we have a contradiction. Therefore, any $\lambda_{i}=0$ for all
$i\in\mathcal{I}\backslash\mathcal{I}_{F}$ and $t=\sum_{i\in\mathcal{I}_{F}}\lambda_{i}f_{i}$
which also implies that $\mathcal{I}_{F}\ne\emptyset$. 

We have just shown that if $t\in F$ then $t$ can be written as a
conical combination of the $f_{i}'s$ in $\mathcal{I}_{F}$. Let us
show the converse. Indeed, for any set of non-negative real numbers
$\left(\lambda_{i},i\in\mathcal{I}\right)$ with $\sum_{i\in\mathcal{I}_{F}}\lambda_{i}f_{i}\in F$
we have $c^{T}\left(\sum_{i\in\mathcal{I}_{F}}\lambda_{i}f_{i}\right)=\sum_{i\in\mathcal{I}_{F}}\lambda_{i}\left(c^{T}f_{i}\right)=0$.
\end{proof}
A simple corollary of is that there is always one and only one face
of $C^{p}$ that contains $t$ in its relative interior, provided
that $t\ne0$. Let us note this formally. 
\begin{cor}
\label{cor:3.2} If $t\in C^{p}$ and $t\ne0$ then there is a unique
face of $C^{p}$ containing $t$ in its relative interior.
\end{cor}
The next theorem pertains to determining the dimension of a face $F$,
which is the dimension of $\mathrm{aff\left(F\right)}$. Henceforth,
for a given face $F$ of $C^{p}$ we define$\mathcal{\mathcal{I}}_{F}=\left\{ i\in\mathcal{I}:f_{i}\in F\right\} $.
\begin{thm}
\label{thm:3.3} If $F$ is a face of $C_{p}$ of dimension at least
one then $\mathrm{aff(}F)=\mathrm{span}\left\{ f_{i},i\in\mathcal{I}_{F}\right\} $.\end{thm}
\begin{proof}
If $x\in\mathrm{aff}\left\{ F\right\} $ then there exist real numbers
$\lambda_{1},\lambda_{2},...,\lambda_{k}$ and points $x_{1},x_{2},...,x_{k}\in F$
such that $\sum_{j=1}^{k}\lambda_{j}=1$ and $x=\sum_{j=1}^{k}\lambda_{j}x_{j}$.
Since $x_{j}\in F$ then $x_{j}=\sum_{i\in\mathcal{\mathcal{I}}_{F}}\alpha_{ij}f_{i}$
for some non-negative real numbers $\alpha_{ij},i\in\mathcal{I}_{F}$.
Therefore, $x=\sum_{j=1}^{k}\lambda_{j}\left(\sum_{i\in\mathcal{\mathcal{I}}_{F}}\alpha_{ij}f_{i}\right)\in\mathrm{span}\left\{ f_{i},i\in\mathcal{I}_{F}\right\} $.
Conversely, if $x\in\mathrm{span}\left\{ f_{i},i\in\mathcal{I}_{F}\right\} $
then $x=\sum_{i\in\mathcal{I}_{F}}\lambda_{i}f_{i}$ for some real
numbers $\lambda_{i},i\in\mathcal{I}_{F}$. Since $0\in F$ we can
write
\[
x=\left(1-\sum_{i\in\mathcal{I}_{F}}\lambda_{i}\right)0+\sum_{i\in\mathcal{I}_{F}}\lambda_{i}f_{i}
\]
which is an affine combination of points in $F$. Therefore, $x\in\mathrm{aff(}F)$.

Since there are $d$ linearly independent $f_{i}'s$, it follows that
the dimension of $C^{p}$ is $d$. \end{proof}
\begin{thm}
\label{thm:3.4} If $F$ is a face of $C_{p}$ of dimension at least
one then the extreme rays of $F$ are the $f_{i}'s$ that belong to
that face. \end{thm}
\begin{proof}
Suppose that $x\in F$ which implies that $x=\sum_{i\in\mathcal{I}_{F}}\lambda_{i}f_{i}$
for some non-negative real numbers $\lambda_{i},i\in\mathcal{I}_{F}$
with at least one $\lambda_{i}>0$. If $x$ is an extreme ray then
we must have exactly one $\lambda_{i}>0$. For if not, then $x$ would
be a conical combination of two linearly independent vectors (none
of the $f_{i}'s$ are scalar multiples of one another). But then $x=\lambda_{i}f_{i}$
for some $i$. 

For a given, $f_{j}\in F$ we need to show that $f_{j}$ is an extreme
ray. Suppose that this is not the case. Then $f_{j}$ can be written
as a conic combination of two vectors $x_{1},x_{2}\in C^{p}$ where
$x_{1}\ne kx_{2}$for some $k>0$. That is, $f_{j}=\lambda_{1}x_{1}+\lambda_{2}x_{2}$
for some $\lambda_{1}\lambda_{2}>0$. But, by theorem \ref{thm:3.1},
$x_{1}=\sum_{i\in\mathcal{I}_{F}}\alpha_{i1}f_{i}$ and $x_{2}=\sum_{i\in\mathcal{I}_{F}}\alpha_{i2}f_{i}$
so that 
\begin{eqnarray*}
f_{j} & = & \lambda_{1}\sum_{i\in\mathcal{I}_{F}}\alpha_{i1}f_{i}+\lambda_{2}\sum_{i\in\mathcal{I}_{F}}\alpha_{i2}f_{i}\\
f_{j} & = & \sum_{i\in\mathcal{I}_{F},i\ne j,\alpha_{i}'>0}\left(\alpha_{i}'\right)f_{i}
\end{eqnarray*}
Recalling that all the $f_{i}'s$ are distinct binary 0/1 vectors
we must have a contradiction.
\end{proof}
The following corollary of theorem \ref{thm:2.2} is due to \cite{Feinberg2}. 
\begin{cor}
\label{cor:3.5} The MLE exists if and only if\textup{ $t\in\mathrm{ri}\left(C^{p}\right)$.}\end{cor}
\begin{proof}
Suppose that the MLE exists. Then by theorem \ref{thm:2.2}, there
exists a $y$ such that $X^{T}y=0$ and $y+n>0$. But then $t=X^{T}\left(y+n\right)$
where $y+n>0$ and hence $t\in\mathrm{ri}\left(C^{p}\right)$. Now
suppose $t\in\mathrm{ri}\left(C^{p}\right)$. By definition, there
exists an $y>0$ such that $X^{T}y=t=X^{T}n$. But then $X^{T}\left(y-n\right)=0$
and $n+(y-n)>0$ and the MLE exists (by theorem \ref{thm:2.2}). 
\end{proof}

\section{An algorithm to determine $F$}

We have seen in section 3, in particular, corollary \ref{cor:3.2},
that there is a unique face $F$ of $C^{p}$ containing the sufficient
statistic $t=X^{T}n$ in its relative interior. We turn now to finding
$F$; for the MLE exists if and only if $F=C^{p}$. With $\mathcal{I}_{F}=\left\{ i\in\mathcal{I}:f_{i}\in F\right\} $
we let $X_{F}$ be an $\mathcal{I}_{F}\times J$ matrix with rows
$f_{i},i\in\mathcal{I}_{F}$. By theorem \ref{thm:3.1}, we know that
$F=\mathrm{cone}\left\{ f_{i},i\in\mathcal{I}_{F}\right\} $ and by
theorem \ref{thm:3.3}, the dimension of $F$ is $d_{F}=\mathrm{rank}\left(X_{F}\right)$.
Equipped with the next result, we can take an approach similar to
\cite{Geyer} and \cite{Feinberg2}, and find $\mathcal{I}_{F}$ by
solving a sequence of linear programs.
\begin{thm}
\label{thm:4.1} Any $a\ge0$ in $R^{\mathcal{I}}$ such that $t=X^{T}a$
must have $a(i)=0$ for any $i\in\mathcal{I\backslash\mathcal{I}}_{F}$.\end{thm}
\begin{proof}
Since $F$ is a face of $C^{p}$ then it is of the form $F=C^{p}\cap\left\{ x\in R^{p}:c^{T}x=0\right\} $
where $c^{T}x<0$ for $x\in C^{p}\backslash F$. Suppose that $t=X^{T}a=\sum_{i\in\mathcal{I}}a(i)f_{i}$
and $a(i)>0$ for some $i\in\mathcal{I}_{F}$. Then
\begin{eqnarray*}
0 & = & \sum_{i\in\mathcal{I}_{F}}a(i)\left(c^{T}f_{i}\right)+\sum_{i\in\mathcal{I}\backslash\mathcal{I}_{F}}a(i)\left(c^{T}f_{i}\right)\\
 & = & \sum_{i\in\mathcal{I}_{F}}a(i)\left(c^{T}f_{i}\right)\\
 & < & 0
\end{eqnarray*}
which is a contradiction. 
\end{proof}
We now present an algorithm to find $\mathcal{I}_{F}$, which we call
the facial set, and show that it works. The algorithm requires solving
a sequence of linear programs that get progressively simpler until
the problem is solved. Let $\mathcal{I}_{0}=\left\{ i\in\mathcal{I}:n(i)=0\right\} $
and $\mathcal{I}_{+}=\left\{ i\in\mathcal{I}:n(i)>0\right\} $. Before
we begin, note that theorem \ref{thm:4.1} applies to $n$ and $\hat{m}$
since $X^{T}\hat{m}=X^{T}n=t$ so that $\mathcal{I}_{+}\subseteq\mathcal{I}_{F}=\left\{ i\in\mathcal{I}:\hat{m}>0\right\} $
or, in other words, $n(i)>0\Rightarrow i\in\mathcal{I}_{F}$ and $i\in\mathcal{I}_{F}\Longleftrightarrow\hat{m}(i)>0$. 
\begin{algorithm}[A repeated linear programming algorithm to find $\mathcal{I}_{F}$]
\emph{\label{alg:4.2} }

\textcolor{white}{\emph{.}}\emph{}\\
\emph{Input: The sufficient statistic $t$ }\\
\emph{Output: The facial set $\mathcal{I}_{F}$.}\end{algorithm}
\begin{enumerate}
\item Set $A=\mathcal{I}_{0}.$ If $A$ is empty then set $\mathcal{I}_{F}=\mathcal{I}\backslash A=\mathcal{I}.$\emph{
}STOP \\

\item Solve the linear program (LP)
\begin{eqnarray}
\mathrm{max} &  & z=\sum_{i\in A}a(i)\nonumber \\
s.t. &  & X^{T}a=t\label{4.1}\\
 &  & a\ge0\nonumber 
\end{eqnarray}

\item If the optimal objective value is $z=0$ then set $\mathcal{I}_{F}=\mathcal{I}\backslash A$.
STOP.
\item Let $a$ be a feasible solution to the LP. For any $i\in A$ such
that $a(i)>0$ remove that index from $A$. Repeat for all feasbile
solutions available (or even just the optimal solution).\\

\item If $A$ is empty, then $\mathcal{I}_{F}=\mathcal{I}\backslash A=\mathcal{I}$.
\emph{STOP}. Otherwise return to\emph{ }STEP 2.
\end{enumerate}
Keeping theorem \ref{thm:4.1} in mind, let us now show that the algorithm
works. If $A$ is empty to being with then it is clear that $t=X^{T}n\in\mathrm{ri}\left(C^{p}\right),F=C^{p},$
and $\mathcal{I}_{F}=\mathcal{I}\backslash A=\mathcal{I}$. The algorithm
terminates at STEP 1. Suppose that $A$ is not empty to start. Observe
that, at any given iteration, the set $\mathcal{I}\backslash A$ is
the set of $i\in\mathcal{I}$ such that we have found some $a\ge0$
in $R^{\mathcal{I}}$ with $X^{T}a=t$ that has $a(i)>0$. Next, observe
that the algorithm terminates if $A$ is empty or $A$ is nonempty
and the optimal objective value is $z=0$. If $A$ is empty then for
all $i\in\mathcal{I}$ we have found some $a\ge0$ in $R^{\mathcal{I}}$
with $x^{T}a=t$ that has $a(i)>0.$ It must be that $\mathcal{I}_{F}=\mathcal{I}$.
If $A$ is non-empty and the optimal objective value is $z=0$ then
every $a\ge0$ such that $X^{T}a=t$ has $a(i)=0$ for $i\in A$.
The result is that $A=\mathcal{I\backslash I}_{F}$ and $\mathcal{I}_{F}=\mathcal{I}\backslash A$.
In all cases we have found $\mathcal{I}_{F}$. 

Now, suppose that $y\in R^{\mathcal{I}}$ such that $y(i)>0\iff n(i)>0$.
Then $t'=X^{T}y=\sum_{i\in\mathcal{I}}y(i)f_{i}$ and $t=X^{T}n=\sum_{i\in\mathcal{I}}n(i)f_{i}$
belong to the same face $F$. Therefore, it is the location of the
zero cells in $n$ that determines $F$ as opposed to the magnitude
of the nonzero entries. For this reason, it is simplest to let 
\[
y(i)=\begin{cases}
1 & n(i)>0\\
0 & n(i)=0
\end{cases}
\]
and find $\mathcal{I}_{F}$ using $t'$. 

For models where $m$ is Markov with respect to a decomposable graph
$G$ (or extended MLE) a closed form expression for the MLE exists.
Theoretically, for such models, one need not resort to linear programming
to find $\mathcal{I}_{F},$ since $\hat{m}$ can be computed exactly.
Once this is done we know that $\mathcal{I}_{F}=\left\{ i\in\mathcal{I}:\hat{m}(i)>0\right\} $.
Practically speaking though, it is simpler to use algorithm \ref{alg:4.2}
for any model without first determining decomposability. 

We now give an example applying algorithm \ref{alg:4.2} 
\begin{example}
Consider a $3\times3\times3$ table for variables $a,b,c$ with counts:\\
\\
\begin{tabular}{|c|c|c|}
\hline 
0 & 1 & 0\tabularnewline
\hline 
0 & 1 & 1\tabularnewline
\hline 
1 & 1 & 1\tabularnewline
\hline 
\end{tabular}%
\begin{tabular}{|c|c|c|}
\hline 
1 & 1 & 1\tabularnewline
\hline 
1 & 1 & 1\tabularnewline
\hline 
1 & 0 & 0\tabularnewline
\hline 
\end{tabular}%
\begin{tabular}{|c|c|c|}
\hline 
1 & 1 & 1\tabularnewline
\hline 
1 & 1 & 1\tabularnewline
\hline 
1 & 0 & 0\tabularnewline
\hline 
\end{tabular}\\
\\
\\
and the model $[ab][bc][ac]$. Suppose $\mathcal{I}_{a}=\mathcal{\mathcal{I}}_{b}=\mathcal{I}_{c}=\{1,2,3\}$.
The leftmost array corresponds to $c=1$, the middle to $c=2,$ and
the rightmost to $c=3$. Let us apply our the algorithm to find $F$
for this data set. We begin, at STEP 1, by setting
\[
A=\mathcal{I}_{0}=\left\{ 111,131,211,322,332,323\right\} 
\]
where, by an abuse of notation, $131$, refers to cell $i=(1,3,1)$
with $a=1,b=3,c=1$. Since $A$ is nonempty we proceed to STEP 2.
The optimal solution to (\ref{4.1}) has $a(131)>0$ and so we remove
the cell $131$ from $A$ to get
\[
A=\left\{ 111,211,322,332,323,333\right\} 
\]
and return to STEP 2. Resolving (1) we find that, this time, the optimal
objective value is 0. At this point we set $\mathcal{I}_{F}=\mathcal{I}\backslash A=\mathcal{I}_{+}\cup\{131\}$
and the algorithm terminates. The dimension of $F$ is $\mathrm{rank}\left(X_{F}\right)$,
which in this case, is 18. 

In the remainder of this section, we proceed to show that when $F\subset C$,
maximum likelihood estimation can proceed almost as usual conditional
on $t\in F$. Recall once more, the likelihood as a function of $m$.

\begin{eqnarray*}
L(m) & = & \prod_{i\in\mathcal{I}}\exp\left(-m(i)\right)m(i)^{n(i)}\\
 & = & \prod_{i\in\mathcal{I}_{F}}\exp\left(-m(i)\right)m(i)^{n(i)}\prod_{i\in\mathcal{I}\backslash\mathcal{I}_{F}}\exp\left(-m(i)\right)m(i)^{n(i)}
\end{eqnarray*}
Since $n(i)=0$ for $i\in\mathcal{I}\backslash\mathcal{I}_{F}$, then
\begin{eqnarray*}
L(m) & = & \prod_{i\in\mathcal{I}_{F}}\exp\left(-m(i)\right)m(i)^{n(i)}\\
 & = & \exp\left(\sum_{i\in\mathcal{I}_{F}}n(i)\log m(i)-\sum_{i\in\mathcal{I}_{F}}m(i)\right)\\
 & = & \exp\left(\left\langle n_{F},\log m_{F}\right\rangle -\sum_{i\in\mathcal{I}_{F}}m(i)\right)
\end{eqnarray*}
where $n_{F}=\left(n(i),i\in\mathcal{I}_{F}\right)$ and $m_{F}=\left(m(i),i\in\mathcal{I}_{F}\right)$.
Let $\mathcal{M}_{F}$ be the linear span of the columns of $X_{F}$.
Then $\hat{m}$ satisfies $\hat{m}(i)=0$, $i\in\mathcal{I}\backslash\mathcal{I}_{F}$
and
\begin{equation}
\hat{m}_{F}=\left(\hat{m}(i),i\in\mathcal{I}_{F}\right)=\mathrm{argsup}{}_{\log m_{F}\in\mathcal{M}_{F}}\exp\left(\left\langle n_{F},\log m_{F}\right\rangle -\sum_{i\in\mathcal{I}_{F}}m(i)\right)\label{eq:4.2}
\end{equation}
The conditional density of $n$ given $t\in F$ is 
\begin{eqnarray*}
P(n(i),i\in\mathcal{I}|t\in F) & = & P(n(i),i\in\mathcal{I}|n(i)=0,i\in\mathcal{I}\backslash\mathcal{I}_{F})\\
 & = & P(n(i),i\in\mathcal{I}_{F})\\
 & = & \prod_{i\in\mathcal{I}_{F}}\exp\left(-m(i)\right)m(i)^{n(i)}\\
 & = & \exp\left(\left\langle n_{F},\log m_{F}\right\rangle -\sum_{i\in\mathcal{I}_{F}}m(i)\right)
\end{eqnarray*}
which is the same as (\ref{eq:4.2}). Therefore, when $F\subset C$,
the MLE of $m$ can be computed by conditional on $t\in F$. Practically
speaking, this means that we can treat the zeros in the cells $i\in\mathcal{I}\backslash\mathcal{I}_{F}$
as structural zeros rather than sampling zeros. Now, if $t\in F$
and $d_{F}=\mathrm{rank}$$\left(X_{F}\right)<d$ then $X_{F}$ is
not of full rank and the model $\log m_{F}=X_{F}\theta$ is over-parametrized;
only $d_{F}$ log-linear parameters will have finite estimates. We
can partially fit the model by selecting $d_{F}$ linearly independent
columns of $X_{F}$, constructing a new design matrix $X_{F}^{*}$,
and fitting the model $\log m_{F}=X_{F}^{*}\theta_{F}$. The new parameter
vector $\theta_{F}$ will contain $d_{F}$ components of $\theta$
which can be estimated. Estimates and standard errors of $m_{F}$
and $\theta_{F}$ can then be obtained as usual. When the contingency
table is not too sparse, and large sample $\chi^{2}$ goodness of
fit statistics are appropriate, the correct degrees of freedom is
$|\mathcal{I}_{F}|-d_{F}$ \cite{Feinberg2}. It is an open research
question whether the Bayesian Information Criterion \cite{Schwarz}
for comparing models should be corrected from $\hat{l}-\frac{d}{2}\log N$
to $\hat{l}-\frac{d_{F}}{2}\log N$ or to something else when $F\subset C$.
\end{example}

\section{The eMLEloglin package}

The main virtue of the eMLEloglin package is the ability to compute
the facial set $F$ for a given log-linear model and data set. It
does this using algorithm \ref{alg:4.2} described above. The required
linear programs are solved using the lpSolveAPI R package. If $F\subset C$,
then a modified contingency table can be constructed, where cells
in $\mathcal{I}\backslash\mathcal{I}_{F}$ are deleted, and passed
to the GLM package to obtain maximum likelihood estimates. The GLM
package will automatically identify a subset of the parameters that
can be estimated and the correct model dimension. We now look at two
examples.
\begin{example}
Consider the following $2\times2\times2$ contingency table for variables
$a,b,c$ with counts:\\
\\
\begin{tabular}{|c|c|}
\hline 
0 & 1\tabularnewline
\hline 
1 & 1\tabularnewline
\hline 
\end{tabular}%
\begin{tabular}{|c|c|}
\hline 
1 & 1\tabularnewline
\hline 
1 & 0\tabularnewline
\hline 
\end{tabular}\\
\\
and the model $[ab][bc][ac]$. Suppose $\mathcal{I}_{a}=\mathcal{\mathcal{I}}_{b}=\mathcal{I}_{c}=\{1,2\}$.
The leftmost array corresponds to $c=1$, and the rightmost array
to $c=2$. This is one of the earliest known examples identified where
the MLE does not exist \cite{Haberman}. We now show how to compute
$F$ using the eMLEloglin package. We first create a contingency table
to hold the data:\begin{verbatim}
> x <- matrix(nrow = 8, ncol = 4) 
> x[,1] <- c(0,0,0,0,1,1,1,1) 
> x[,2] <- c(0,0,1,1,0,0,1,1) 
> x[,3] <- c(0,1,0,1,0,1,0,1) 
> x[,4] <- c(0,1,2,1,4,1,3,0) 
> colnames(x) = c("a", "b", "c", "freq") 
> x
> x a b c freq 
  1 0 0 0 0 
  2 0 0 1 1 
  3 0 1 0 2 
  4 0 1 1 1 
  5 1 0 0 4 
  6 1 0 1 1
  7 1 1 0 3
  8 1 1 1 0 
\end{verbatim}We can then use the $\mathtt{facial\_set}$ function:\begin{verbatim}
> f <- facial_set (data = x, formula = freq ~ a*b + a*c + b*c) 
> f
$formula 
freq ~ a*b + a*c + b*c

$model.dimension 
[1] 7

$status 
[1] "Optimal objective value 0"

$iterations 
[1] 1

$face   
  a b c freq facial_set 
1 0 0 0    0          0 
2 0 0 1    1          1 
3 0 1 0    2          1 
4 0 1 1    1          1 
5 1 0 0    4          1 
6 1 0 1    1          1 
7 1 1 0    3          1 
8 1 1 1    0          0

$face.dimension 
[1] 6

$maxloglik [1] 
-1.772691 
\end{verbatim}The output begins by giving the model formula and the original dimension.
Under Poisson sampling the model of no three-way interaction has 7
free parameters. The line mentioning status is for debugging purposes
to know how the algorithm terminated. For this example, algorithm
\ref{alg:4.2} terminated when an optimal objective value of $z=0$
was found. The next line indicates that the algorithm required only
one iteration to find $F$. The table in $\mathrm{\mathtt{f2\$face}}$
is probably the most important output. It indicates that 
\begin{eqnarray*}
\mathcal{I}_{F} & = & \left\{ 001,010,011,100,101,110\right\} \\
\mathcal{I}\backslash\mathcal{I}_{F} & = & \left\{ 000,111\right\} 
\end{eqnarray*}
The implication of the fact that here $\mathcal{I}_{F}\ne\mathcal{I}$
is that the dimension of $F$ is 6 which we see under \$face.dimension.
Since $|\mathcal{I}_{F}|=d_{F}=6$, the model is effectively saturated
and the fitted values are the same as the observed values. We can
see this by passing the data with the cells in $\mathcal{I\backslash}\mathcal{I}_{F}$
removed to the glm function. 
\end{example}
\begin{verbatim}
> fit <- glm (formula = freq ~ a * b + a * c+ b * c, 
              data = x[as.logical(f2$face$facial_set),] 

> fit Call:  glm(formula = freq ~ a * b + a * c + b * c, 
                 data = x[as.logical(f2$face$facial_set),      ])

Coefficients: (Intercept)            a            b            c 
            	   2.000e+00    2.000e+00   -1.479e-15   -1.000e+00
        a:b          a:c        b:c       
 -1.000e+00   -2.000e+00	       NA  

Degrees of Freedom: 5  Total (i.e. Null);  0 Residual 
Null Deviance:	  8  Residual Deviance: 6.015e-30 	
AIC: -383.4 

> fit$fitted.values 
2 3 4 5 6 7  
1 2 1 4 1 3  
\end{verbatim}As expected, one parameter is not able to be estimated; and R handles
this automatically. Note that the residual degrees of freedom is correctly
calculated to be $|\mathcal{I}_{F}|=d_{F}=0$. Let us work through
a larger example now with the Rochdale data.
\begin{example}
The eMLEloglin package includes a sparse dataset from the household
study at Rochdale referred to in \cite{Whittaker}. The Rochdale data
set is a contingency table representing the cross classification of
665 individuals according to 8 binary variables. The study was conducted
to elicit information about factors affecting the pattern of economic
life in Rochdale, England. The variables are as follows: a. wife economically
active (no, yes); b. age of wife >38 (no, yes); c. husband unemployed
(no, yes); d. child$\le4$ (no, yes); e. wife's education, highschool+
(no, yes); f. husband's education, highschool+ (no, yes); g. Asian
origin (no, yes); h. other household member working (no, yes). The
table is sparse have 165 counts of zero, 217 counts with at most three
observations, but also a few large counts with 30 or more observations.
The Rochdale data comes preloaded with the package. Suppose we are
interested in the model $\mathrm{|ad|ae|be|ce|ed|acg|dg|fg|bdh|}$
which is the model with the highest corrected BIC for this data set.
We give a list of the top models by corrected and uncorrected BIC
below for this data set. The required R code to find the facial set
for this model is:
\end{example}
\begin{verbatim}
data(rochdale)
f <- facial_set  (data = rochdale, 
                  formula = freq ~ a*d + a*e + b*e + c*e + e*f + 
                                 a*c*g + d*g + f*g + b*d*h)
\end{verbatim}From the output we see that the model lies on a face of dimension
22. Since the original model dimension is 24, two parameters will
not be estimable. Given the sparsity of the table, a goodness of fit
test would not be appropriate. The fitted model can be obtained from
the GLM function with the code:\begin{verbatim}
fit <- glm (formula = freq ~ a*d + a*e + b*e + c*e + e*f + 
                             a*c*g + d*g + f*g + b*d*h)
               data = rochdale[as.logical(f2$face$facial_set),]) 
\end{verbatim}The GLM function automatically determines that $\theta_{acg}$ and
$\theta_{bdh}$ can not be estimated; which is because the $\mathrm{acg}$
and $\mathrm{bdh}$ margins are both zero. The residual degrees of
freedom is correctly calculated at $|\mathcal{I}_{F}|-d_{F}=196-22=174.$ 

Since the Rochdale dataset seems to be of some interest recently we
give the top five models in terms of corrected and the usual BIC (abbreviated
cBIC and BIC in the tables, respectively) found while exploring models
using the MC3 algorithm.\\
\\
\begin{tabular}{|c|c|c|c|}
\hline 
 & cBIC & Model Dim. & Face Dim.\tabularnewline
\hline 
$\mathrm{|ad|ae|be|cd|ef|acg|dg|fg|bdh|}$ & 985.3 & 24 & 22\tabularnewline
\hline 
$\mathrm{|ad|ae|be|ce|cf|ef|acg|dg|fg|bdh|}$ & 985.2 & 25 & 23\tabularnewline
\hline 
$\mathrm{|ad|ae|be|ce|cf|df|ef|acg|dg|fg|bdh}$ & 984.4 & 26 & 24\tabularnewline
\hline 
$\mathrm{|ad|ae|be|ce|df|ef|acg|dg|fg|bdh|}$ & 984.3 & 25 & 23\tabularnewline
\hline 
$\mathrm{|ac|ad|ae|be|ce|ef|ag|cg|dg|fg|bdh}$ & 984.0 & 23 & 22\tabularnewline
\hline 
\end{tabular}\\
\\
\\
\begin{tabular}{|c|l|c|c|}
\hline 
Model & BIC & Model Dim. & Face Dim.\tabularnewline
\hline 
$\mathrm{|ac|ad|bd|ae|be|ce|ef|ag|cg|dg|fg|bh|dh|}$ & 981.3 & 22 & 22\tabularnewline
\hline 
$\mathrm{|ac|ad|bd|ae|be|ce|cf|ef|ag|cg|dg|fg|bh|dh|}$ & 981.1$^{*}$ & 23 & 23\tabularnewline
\hline 
$\mathrm{|ac|ad|ae|be|ce|ef|ag|cg|dg|fg|bdh|}$ & 980.7 & 23 & 22\tabularnewline
\hline 
$\mathrm{|ac|ad|ae|be|ce|cf|ef|ag|cg|dg|fg|bdh|}$ & 980.5$^{**}$ & 24 & 23\tabularnewline
\hline 
$\mathrm{|ac|ad|bd|ae|be|ce|ef|ag|cg|dg|fg|bh|}$ & 980.4 & 21 & 21\tabularnewline
\hline 
\end{tabular}\\
\\
\\
With the exception of $\mathrm{cf}$ and the three factor interactions,
the model with the highest uncorrected BIC is the model identified
by \cite{Whittaker} who, in any case, limited himself to considering
at most two-factor interactions because of the sparsity of the table.
Whittaker fit the all two factor interaction model, and then deleted
the terms that we non-significant and arrived at the model $\mathrm{|ac|ad|bd|ae|be|ce|cf|ef|ag|cg|dg|fg|bh|dh}$
marked by an asterisk ({*}) above.

We note that $\mathrm{bdh}$ interaction has also been identified
by \cite{Dobra1}. The model they selected is marked with ({*}{*})
above. Using Mosaic plots, \cite{Hofmann} also observed that there
is a strong hint of $\mathrm{bdh}$ interaction. The dataset was also
analyzed in \cite{Dobra2}.

\bibliographystyle{plain}
\bibliography{mybib}

\end{document}